\documentclass[12pt]{article}

\usepackage{graphicx} 
\usepackage[utf8]{inputenc}
\usepackage{fullpage, amsmath, amsfonts, amsthm, color, graphicx, amssymb, nicefrac}
\usepackage{algpseudocode}
\usepackage{hyperref}

\usepackage{authblk}
\usepackage{comment}

\newtheorem{theorem}{Theorem}
\newtheorem{lemma}{Lemma}
\newtheorem{corollary}{Corollary}

\theoremstyle{definition}
\newtheorem{algorithm}{Algorithm}
\newtheorem{definition}{Definition}
\newtheorem{remark}{Remark}

\newcommand{\Tr}{\text{Tr}}

\graphicspath{ {./images/} }

\title{An improved Quantum Max Cut approximation\\ via Maximum Matching}
\date{}

\author[1]{Eunou Lee\thanks{\href{mailto:eunoulee@kias.re.kr}{eunoulee@kias.re.kr}}}
\author[2]{Ojas Parekh\thanks{\href{mailto:odparek@sandia.gov}{odparek@sandia.gov}}}
\affil[1]{Korea Institute for Advanced Study, Seoul, South Korea }
\affil[2]{Sandia National Laboratories, Albuquerque, USA}

\begin{document}

\maketitle

\begin{abstract}
    Finding a high (or low) energy state of a given quantum Hamiltonian is a potential area to gain a provable and practical quantum advantage.
    A line of recent studies focuses on Quantum Max Cut, where one is asked to find a high energy state of a given antiferromagnetic Heisenberg Hamiltonian.
    In this work, we present a classical approximation algorithm for Quantum Max Cut that achieves an approximation ratio of 0.595, outperforming the previous best algorithms of Lee \cite{Lee22} (0.562, generic input graph) and King \cite{King22} (0.582, triangle-free input graph).
    The algorithm is based on finding the maximum weighted matching of an input graph and outputs a product of at most 2-qubit states, which is simpler than the fully entangled output states of the previous best algorithms. 
\end{abstract}

\bigskip
\section*{Introduction}
A quantum optimization problem seeks to compute the maximum (or minimum) of a function that is defined over the $n$-qubit Hilbert space.
In a restricted case where the function is a sum of $k$-qubit Hamiltonians, it is well known that the problems are in general QMA-hard, i.e. hard to solve to an inverse polynomial precision even with a quantum computer \cite{GHLS15}.
One way to cope with the computational hardness is to try to find good approximate solutions.

Quantum Max Cut (QMC) has served as a benchmark problem to develop ideas for quantum Hamiltonian approximation.
It has a simple definition, has a good physical motivation, namely the antiferromagnetic Heisenberg model, and extends the well-studied classical Max Cut problem.
The task of QMC, given a positive weighted graph $G = (V,E,w),$ is to find a (description of a) maximum energy state for the Hamiltonian
\begin{align*}
    H = \sum_{(i,j) \in E} w_{ij} (I - X_i X_j - Y_i Y_j - Z_i Z_j)/4,
\end{align*}
where $X_i$ is the Pauli matrix $X$ on qubit $i$ and identity on the rest. The matrices, $Y_i, Z_i$ are similarly defined.

Most of the existing approximation algorithms for QMC follow the framework of the seminal Goemans-Williamson algorithm \cite{GW95}. The problem is first relaxed to a semidefinite program (SDP), then the SDP is solved in polynomial time, and finally, the SDP solution is rounded to a feasible solution of the original problem.
Since QMC is a maximization problem over the $n$-qubit Hilbert space, the rounded solution should be an $n$-qubit quantum state instead of an $n$-bit string as for Max Cut. 
Assuming that we follow the Goemans-Williamson framework to approximate QMC, there are still three design choices: 1) which SDP to relax the original problem to, 2) which subset of quantum states (ansatz) to round the SDP solution to, and 3) how to round the SDP solution to an ansatz state.  For Max Cut, optimal choices (up to the Unique Games Conjecture) of classical analogues of the above are known; however, these remain unsettled for QMC.  

The Quantum Lasserre SDP hierarchy \cite{L02, ANP08} is a sequence of SDPs that upper bounds the maximum energy of a given quantum Hamiltonian. The Quantum Lasserre hierarchy does so by optimizing over pseudo-density operators that are not guaranteed to be positive. The level-$k$ Lasserre SDP includes all valid linear constraints on moments of subsets of at most $k$ qubits. It also includes global constraints characterized by polynomials where each term is tensor product of at most $k$ non-trivial single-qubit Paulis (see e.g.,~\cite{PT22}). 
Hence the SDPs in the hierarchy become tighter as the level increases, eventually converging to the given quantum Hamiltonian problem when $k = n$. The following is a way to view the SDP construction. Fix a quantum state $|\phi\rangle$. 
For an $n$-qubit Pauli matrix $P,$ define $v(P):= P|\phi\rangle.$
Then the energy of $|\phi\rangle$ for an arbitrary Hamiltonian can be expressed as a sum of inner products of these vectors. For example, $\langle \phi |I - X_i X_j - Y_i Y_j - Z_i Z_j  | \phi\rangle = \langle v(I), v(I)\rangle - \langle v(X_i X_j), v(I)\rangle -\langle v(Y_i Y_j), v(I)\rangle -\langle v(Z_i Z_j), v(I)\rangle $.
Additionally, it holds that $\|v(P) \| = 1$, and $\langle v(P_1), v(Q_1)\rangle = \langle v(P_2), v(Q_2)\rangle$ for all $n$-qubit Pauli matrices $P_1, P_2,Q_1, Q_2$ such that $P_1 Q_1 = P_2 Q_2.$
Now forget about $|\phi\rangle$ and maximize the energy expression in terms of $v(P)$'s, while satisfying the inner product relations.

In the following definition, $\mathcal P_k(n)$ is the set of $n$-qubit Pauli matrices with non-trivial terms on up to $k$ qubits.
\begin{definition}[Level-$k$ Quantum Lasserre SDP]
\begin{align*}
    \text{Maximize } \sum_{(i,j) \in E} w_{ij} v(I)\cdot(v(I) - v(X_i X_j) - v(Y_i Y_j) -v(Z_i Z_j))/4 \label{eq:sdp}\tag{S}
\end{align*}
    \begin{align*}
    \text{subject to } 
    & v(P) \in \mathbb R^{|\mathcal P_k(n)|}, & \forall P \in \mathcal P_k(n),    \\
    & v(P)\cdot v(P) = 1, & \forall P\in \mathcal P_k(n), \\
    & v(P_1)\cdot v(Q_1) = v(P_2)\cdot v(Q_2), & \forall P_1, P_2,Q_1, Q_2 \in \mathcal P_k(n) \text{ s.t. }P_1Q_1 = P_2Q_2, \\
    & v(P)\cdot v(Q) = 0, & \forall P,Q \in \mathcal P_k(n)\text{ s.t. } PQ+QP = 0.
\end{align*}
\end{definition}
Every existing QMC approximation algorithm that follows an SDP rounding framework uses a Quantum Lasserre SDP.
The QMC approximation algorithm of Gharibian-Parekh \cite{GP19} employs the level-1 Lasserre, and Parekh-Thompson \cite{PT21-lev2}, Parekh-Thompson \cite{PT22}, Lee \cite{Lee22}, and King \cite{King22} employ the level-2 Lasserre SDP.
More sophisticated SDP hierarchies that are aware of the $SU(2)$ symmetry in the QMC Hamiltonian have also been developed recently \cite{KPTTZ23,CEHKW23}, and such hierarchies are implicitly used in existing QMC approximation algorithms~\cite{PT22}.

Once an SDP relaxation is solved, its solution is rounded to a quantum state.
Current algorithms round SDP solutions to a proper subset (ansatz) of the $n$-qubit Hilbert space.
Gharibian-Parekh \cite{GP19} and Parekh-Thompson \cite{PT22} round to product states, and Parekh-Thompson \cite{PT21-lev2} rounds to a product of 1- and 2-qubit states, inspired by the non-SDP approximation algorithm of Anshu, Gosset, and Morenz~\cite{AGM20}. Lee \cite{Lee22} and King \cite{King22} round to $n$-qubit entangled states.

In this paper, we introduce a simple classical approximation algorithm that solves the level-2 Lasserre SDP and rounds to either a product of 1-qubit states or a product of 1- and 2-qubit states. The former is obtained using the Gharibian-Parekh product state rounding algorithm. The latter is obtained by solving the Maximum Weight Matching problem in the weighted input graph on which the QMC Hamiltonian is defined, and this does not depend on the SDP solution at all. This distinguishes and drastically simplifies our algorithm relative to previous SDP rounding approaches, which crucially use information from a level-2 SDP solution to produce entangled states.

We show that the approximation ratio of our algorithm is 0.595, which improves the previous best algorithms for general graphs (Lee \cite{Lee22} with a ratio of 0.562), and triangle-free graphs (King \cite{King22} with a ratio of 0.582). 

\paragraph{Quantum optimization.}
Another issue that we are concerned with is the definition of quantum optimization problems. 
A common way to define a quantum optimization problem in the literature is to define an energy function on $n$ qubits and then seek a maximum-energy state with respect to the function. What does it mean for a classical algorithm to solve this problem? If we restrict classical algorithms to outputting basis states, then the above quantum optimization reduces to a classical one.  A more relaxed and common approach is to only ask for a ``description'' of a quantum state.  This definition accommodates a broader family of algorithms that output a description of an entangled state, such as our and other previous algorithms for approximating QMC.
Now the issue is that the meaning of the word ``description'' is vague:
an output state of any quantum algorithm has a classical description, namely the quantum algorithm itself written on paper.
Therefore any quantum algorithm is a classical algorithm if we accept this definition.
We propose the following definition for a more satisfying notion of a classical algorithm solving a quantum optimization problem.
\begin{definition}
    Given an objective function $f$ that maps an $n$-qubit state to a real number, a pair of polynomial time quantum or classical algorithms $(P,V)$ maximizes $f$ to a value $\nu$ if the following conditions hold:
\begin{enumerate}
    \item 
    \begin{enumerate}
        \item  If $\text{max }f \ge \nu$:  $\exists$ $|w\rangle$ of size polynomial in $n$ such that $V(|w\rangle) =1$ w.p. $\ge 2/3$,
        \item  If $\text{max }f \le \nu - 1/p(n)$ for some polynomial $p$:  $\forall$ polynomial-sized $|w\rangle$, $V(|w\rangle) =1$ w.p. $\le 1/3$ 
    \end{enumerate}
    \item $P$ outputs $|w'\rangle$ such that $V(|w'\rangle)=1 $ w.p. $\ge 2/3.$
\end{enumerate}
\end{definition}
Each of $P$ and $V$ can be either classical or quantum; when $P$ or $V$ is classical, it is assumed that classical states are employed. Therefore according to our definition, we can have a cc-, qc-, cq-, or qq-optimization algorithm for a quantum optimization problem depending on whether each of $P$ and $V$ is classical or quantum.
Only cc-optimization algorithms for QMC are known to us so far.

We argue that finding an optimal qq-optimization algorithm for QMC is a viable path to achieve a \emph{provable and practical quantum advantage}.
Suppose the verifier algorithm $V$ is fixed to be classical. 
Then, from the prover $P$'s side, the task is to find a bit string to convince $V$ that there is a high energy state.
Assuming there is a tight NP-hard upper bound for a classical approximation for QMC (for example by the PCP theorem),
we cannot find a quantum prover that provably gives a greater ratio than all classical provers unless we prove that NP is a proper subset of QMA. 
However, to the best of our knowledge, there are no unexpected complexity theoretic consequences of a qq-approximation having a greater ratio than all cc-approximations. 
Moreover, we already know how to upper bound classical approximation ratios in some cases via the PCP theorem and the Unique Games Conjecture, so we can hope to upper bound the classical ratio for QMC.
The current best upper bound for a cc-approximation of QMC is 0.956 up to plausible conjectures~\cite{HNPTW21}.

\section*{Approximation algorithm}
In all previous algorithms outputting entangled states, the following monogamy of entanglement relation on stars is used crucially. 
\begin{definition}[SDP solution values]
  Let $G= (V= [n],E,w)$ be a weighted graph and let $(v(P))_{P \in \mathcal P_2(n)}$ be a feasible solution to the level-2 Lasserre SDP. Define, for $i,j \in V$, 
    \begin{gather*}
         g_{ij} := \frac{1}{4}v(I)\cdot(v(I) - v(X_i X_j) - v(Y_i Y_j) -v(Z_i Z_j))\\
         h_{ij} := \frac{1}{4}v(I)\cdot(v(I) - v(X_i X_j) - v(Y_i Y_j) -v(Z_i Z_j)) - \frac{1}{2}.
    \end{gather*}
  For  $x \in \mathbb R$, denote $x^+:= \text{max}(x,0)$.  In particular for $i,j \in V$,
    \begin{align*}
        h^+_{ij} := \max(h_{ij}, 0).
    \end{align*}
   The objective function value of the SDP solution is then $\nu := \sum_{(i,j) \in E} w_{ij} g_{ij}$.
\end{definition}

\begin{lemma}[Monogamy of entanglement on a star, \cite{PT21-lev2}] \label{lem:star_monogamy} Given a feasible solution to the level-2 Lasserre SDP on a graph $G=(V,E)$, for any vertex $i \in V$ and any $S\subseteq V$,
\begin{align*}
    \sum_{j \in S} h_{ij} \le \frac{1}{2}.
\end{align*}
In particular,
\begin{align}
    \sum_{j \in N(i)} h^+_{ij} \le \frac{1}{2},  \label{eq:hij_monogamy}
\end{align}
where $N(i)= \{j| (i,j) \in E\}.$
\end{lemma}
The last statement is obtained by taking the set of edges incident to $i$ with positive $h_{ij}$ values as $S$. 

A  matching in $G$ corresponds naturally to a state that earns maximal energy on the Hamiltonian terms corresponding to matched edges. Notice that by Equation~\eqref{eq:hij_monogamy}, $(2h^+_e)_{e \in E}$ forms a fractional matching in the sense that if all these values were either 0 or 1, then we would have a matching that would in turn yield a state.  We can round this fractional matching to a true matching, with a loss in objective value. Since a maximum weight matching has weight at least that of any matching we might round to, we may simply use it instead.  
We solve the Maximum Weight Matching problem and assign the best 2-qubit state, namely the singlet, $(|01\rangle - |10\rangle)/\sqrt{2}$, to matched edges.
To each unmatched qubit, we assign the maximally mixed 1-qubit state.
The resulting $n$-qubit state has maximal energy on matched edges but low energy on edges that are not matched. To address this issue, we also run the product state algorithm of Gharibian-Parekh \cite{GP19}. Below is the complete algorithm.
\begin{algorithm}[Approximation algorithm for Quantum Max Cut] \label{alg}
Given a weighted graph $G = (V,E,w)$ as an input,
\begin{enumerate}
    \item \label{item:product_state} Find a product state as follows:
    \begin{enumerate}
        \item \label{item:solve_SDP} Solve SDP \eqref{eq:sdp} for $k=2$ to get solution vectors $(v(P))_{P\in \mathcal P_2(n)}$.
        \item Sample a random matrix $R$ with dimension $3\times 3 |\mathcal P_2(n)|$ with each element independently drawn from $\mathcal N(0,1)$.
        \item Perform Gharibian-Parekh rounding on each $v_i : = (v(X_i) \| v(Y_i) \| v(Z_i)) /\sqrt 3$ to get
        $u_i:= Rv_i/\|Rv_i\|$.
        \item Let $\rho_1:= \prod_{i\in V}\frac{1}{2}( I + u_{i,1} X_i +u_{i,2} Y_i +u_{i,3} Z_i ) $. 
    \end{enumerate}
    \item \label{item:matching} Find a matching state as follows:
    \begin{enumerate}
        \item Find the maximum weight matching $M:E\rightarrow \{0,1\}$ of $G$, for example via Edmonds Algorithm \cite{Ed65}.
        \item Let $\rho_2 := \prod_{(i,j):  M_{(i,j)} =1}(I -X_i X_j - Y_i Y_j - Z_i Z_j)/4 \prod_{i \in U} I/2$, where $U$ is the set of qubits unmatched by $M$.
    \end{enumerate}
    \item Output whichever of $\rho_1$ and $\rho_2$ that has greater energy.
\end{enumerate}
\end{algorithm}
Using matchings to find a good QMC state is not a new idea. 
Anshu, Gosset, and Morenz \cite{AGM20} introduced the idea of using matchings for QMC approximations, proving that there exists a product of 1- and 2-qubit states with energy at least 0.55 times the maximum QMC energy.
Parekh and Thompson use a level-2 Lasserre solution to identify~\cite{PT21-lev2,PT22} a subgraph on which they find a maximum weight matching; they then output a product state or a product of 1- and 2-qubit states akin to our algorithm, yielding an approximation ratio of 0.533.
To obtain our improvement, we relate a level-2 SDP solution to the value of a matching on the whole input graph, whereas Parekh-Thompson do so for a proper subgraph of the input graph that is obtained from the level-2 SDP solution.    

Our algorithm is much simpler than previous algorithms producing entangled states. It may be surprising that we can establish that this algorithm offers a better approximation guarantee than previous algorithms, including those outputting states with potential global entanglement.  In particular, our algorithm does not even need to solve an SDP to obtain the entangled solution $\rho_2$.  We only use the level-2 Lasserre SDP to argue that a maximum weight matching provides a solution that has reasonably high energy when $\rho_1$ has low energy.  This manifests itself when obtaining the product state $\rho_1$, which must be done so with respect to the level-2 SDP.  Even though $\rho_1$ requires solving the level-2 SDP, it is obtained by only using the vectors, $v(X_i), v(Y_i), v(Z_i)$, corresponding to single-qubit Paulis (i.e., the level-1 part of a level-2 solution). In fact, our algorithm is well defined if we solve the level-1 SDP relaxation instead of the level-2 SDP in Step \eqref{item:solve_SDP}; however, we do not expect approximation factors beyond 0.498 using only the level-1 SDP~\cite{HNPTW21}.    

\paragraph{Strengthening monogamy of entanglement on a star.}
All previous approximation algorithms for QMC outputting entangled states critically rely on monogamy of entanglement on a star (Lemma~\ref{lem:star_monogamy}). The previously best-known algorithms by Lee \cite{Lee22} and King \cite{King22} both start with a good product state and evolve it to an entangled state while respecting these monogamy of entanglement constraints at each vertex.
The key difference in our case is that we directly find a global solution satisfying the monogamy of entanglement constraints. If our algorithm is the optimal way to exploit these constraints, stronger inequalities obtained from the SDP are necessary to deliver a QMC approximation algorithm with a meaningful improvement in approximation ratio.

Even though Monogamy of Entanglement is tight when $h_{ij} = 1/|N(i)|$ for all $j \in N(i)$, it is easy to see that the inequality is far from tight at other points. 
Suppose $h_{ij}= 1/2$, with $(i,j)$ being maximally entangled. Then on a neighbouring edge $(i,k)$, the energy is 1/4 and $h_{ik} = -1/4$.
In this case, the deviation from the upper bound grows linearly as the number of connected edges grows. Parekh and Thompson derived nonlinear monogamy of entanglement inequalities on a triangle to address this issue in obtaining an optimal approximation for QMC using product states.  Their result is captured in Lemma~\ref{lem:triangle_monogamy}.  

\begin{lemma}[Monogamy of entanglement on a triangle, Lemma 1 of \cite{PT22}]  \label{lem:triangle_monogamy}
  Given a feasible solution to the level-2 Lasserre SDP on a graph $G=(V,E)$, for $i,j,k \in V$,
  \begin{gather}
    0 \leq g_{ij} + g_{jk} + g_{ik} \leq \frac{3}{2} \label{eq:triangle_linear_monogamy} \\
    g_{ij}^2 + g_{jk}^2 + g_{ik}^2 \leq 2g_{ij} g_{jk} + 2g_{ij}g_{ik} + 2g_{jk}g_{ik}.  \label{eq:triangle_quadratic_monogamy}
  \end{gather}
\end{lemma}

The presentation of the above lemma is equivalent to that of~\cite{PT22} after a change of variables. From these relations, we obtain a tighter bound on star graphs with 2 edges as stated in the lemma below.
We denote the relation ``convexgamy'' to distinguish it from the linear monogamy relation (Lemma \ref{lem:star_monogamy}), and the resulting relation gives a convex region.

\begin{figure}
  \centering
    \includegraphics[width=0.5\textwidth]{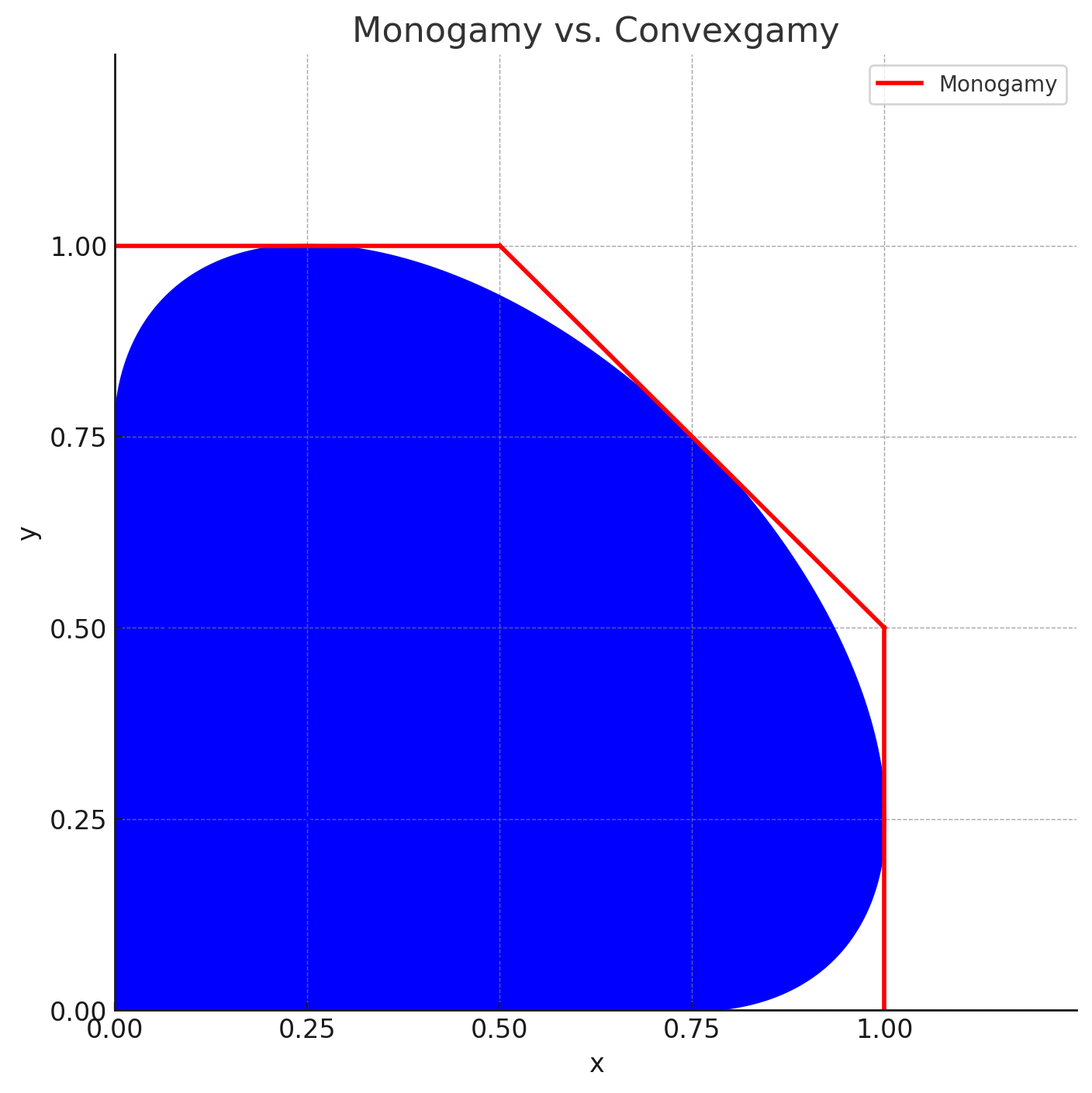}
  \caption{The solid region represents a feasible area characterized by entanglement convexgamy whereas the red line represents the boundary of the feasible region by monogamy of entanglement on a star.}
\label{fig:fig}
\end{figure}

\begin{lemma}[Entanglement convexgamy on 2 edges]
Consider a graph $G = (V= \{1,2,3\},E= \{(1,2),(2,3)\})$. Let $x,y $ be defined by a feasible level-2 Lasserre SDP solution as $x := g_{12}$ and $y := g_{23}$.  Then $(x, y)$ is confined in the region defined by the x-axis, y-axis, and the ellipse touching the x- and y-axis and $x+y=3/2$ as depicted in Figure \ref{fig:fig}. More specifically, the ellipse is defined as $3(x+y-1)^2 + (x-y)^2  = 3/4$. 
\end{lemma}
\begin{proof}
Let $z = g_{13}$. Then $x,y,z\ge 0.$
By Lemma~\ref{lem:triangle_monogamy}, $x^2 + y^2 + z^2 \le 2(xy + yz + zy),$ and $0\le x + y +z \le 3/2.$ The first inequality is equivalent to
\begin{align*}
    \frac{\sqrt{ x^2 + y^2 + z^2 }}{x + y + z} \le \frac{1}{\sqrt 2}.
\end{align*}
It means that if $x+ y +z =c \ge 0,$ then $\sqrt{ x^2 + y^2 + z^2 } \le c/\sqrt{2}$. So $(x,y,z)$ lies the intersection of the plane $x+y+z = c$ and the sphere of radius $c/\sqrt 2$ centered at the origin, which is the incircle of the triangle defined by $x+ y + z\ge c$ in the region $x,y,z\ge 0$. When the circle is projected to the $xy$-plane to give a feasible subset of $(x,y)$, we get the ellipse inscribed in the triangle defined by $(0,0), (c,0), (0,c).$ Because $0\le c \le 3/2$, we prove that a feasible point is in the region defined by the $x$-axis, $y$-axis, and the ellipse touching the $x$- and $y$-axis and $x+y = 3/2$.

The equation of the ellipse follows by solving the inscription condition.
\end{proof}

\bigskip
\section*{Analysis of the algorithm}
In the rest of the paper, we introduce necessary concepts regarding matching theory and bound the approximation ratio of our algorithm.

\begin{theorem}[Linear program for Maximum Weight Matching, \cite{Ed65}]  Given a weighted graph $G= (V, E, w)$, the following linear program gives the value of a maximum weight matching in $G$:
\begin{alignat}{3}
&\text{maximize} &\quad&  \sum_{e \in E} w_e x_e  \label{eq:matching_lp}\tag{M}\\
&\text{subject to} &&  \sum_{j \in N(i)} x_{ij} \leq 1, &\quad& \text{for all } i \in V, \label{eq:matching_vertex}\\ 
& && \sum_{e \in E(S)} x_e \leq \frac{|S|-1}{2}, && \text{for all } S \subseteq V: |S| \text{ odd},  \label{eq:matching_set}\\
& && x_e \geq 0, && \text{for all }e \in E.  \label{eq:matching_nonneg}
\end{alignat}
where $E(S) := \{ (i,j) \in E \mid i,j \in S \}$ for all $S \subseteq V$.
\end{theorem}
The above linear program (LP) cannot be efficiently solved directly since it has an exponential number of constraints; however, algorithms for Maximum Weight Matching, such as Edmonds Algorithm~\cite{Ed65}, obtain the optimal value in polynomial time using insights based on the LP and its dual.  We will need to show that if we are given a solution $(x)_{e \in E}$ that only satisfies \emph{some} of the constraints, then $(\alpha x)_{e \in E}$ satisfies \emph{all} of the constraints for some $\alpha \in (0,1)$.  This will allow us to relate the objective value of the level-2 SDP to the weight of an optimal matching.

\begin{lemma}  \label{lem:matching_approx}
If $(x)_{e \in E}$ satisfies constraints \eqref{eq:matching_vertex},
\eqref{eq:matching_nonneg},  and \eqref{eq:matching_set} for $|S|=3$, then $(\frac{4}{5}x)_{e \in E}$ is feasible for LP \eqref{eq:matching_lp}.
\end{lemma}
\begin{proof}Assume the hypothesis and consider the constraints of \eqref{eq:matching_set}.  We bound the value of the solution $x$ on each such constraint.  For $S \subseteq V$, define $\delta(S) := \{(i,j) \in E \mid |\{i,j\} \cap S| = 1\}$.  Then,
\begin{equation*}
2\sum_{e \in E(S)} x_e \leq 2\sum_{e \in E(S)} x_e + \sum_{e \in \delta(S)} x_e = \sum_{i \in S} \sum_{j \in N(i)} x_{ij} \leq |S|,
\end{equation*}
where the first inequality follows from \eqref{eq:matching_nonneg} and the second from \eqref{eq:matching_vertex}.  So we have
\begin{align*}
\sum_{e \in E(S)} x_e \leq \frac{|S|}{2}, \text{for all } S \subseteq V;
\end{align*}
however, to satisfy \eqref{eq:matching_set}, we need a RHS of $\frac{|S|-1}{2}$ instead of $\frac{|S|}{2}$ for sets $S$ of odd size.  Since the $|S|=3$ case is satisfied by assumption, $(\alpha x)_{e \in E}$ is feasible for LP \eqref{eq:matching_lp}, where
\begin{equation*}
\alpha := \min_{\{s \in \mathbb{Z} \mid s \text{ odd}, s > 3\}} \frac{s-1}{s} = \frac{4}{5}.
\end{equation*}        
\end{proof}

In order to appeal to the above lemma, we will need to show that the energy values arising from the level-2 SDP satisfy the constraints in the hypothesis of the lemma. For this we will rely on monogamy of entanglement on a star and triangle as established in Lemma~\ref{lem:star_monogamy} and the following corollary of Lemma~\ref{lem:triangle_monogamy}, respectively.  

\begin{corollary} \label{cor:triangle_bound}
  Given a feasible solution to the level-2 Lasserre SDP on a graph $G=(V,E)$, for $i,j,k \in V$,
  \begin{align}
    h^+_{ij} + h^+_{jk} + h^+_{ik} \leq \frac{1}{2}.  \label{eq:hij_triangle}
  \end{align}  
\end{corollary}
\begin{proof}
Let t be the number of $(u,v) \in \{(i,j),(j,k),(i,k)\}$ with $h^+_{uv} > 0$.  If $t \leq 1$ then \eqref{eq:hij_triangle} holds since $h^+_{ij} \leq \frac{1}{2}$ for all $i,j \in V$.  If $t \geq 2$ then \eqref{eq:hij_triangle} holds by \eqref{eq:triangle_linear_monogamy}.
\end{proof}

We are now in position to prove our main result.

\begin{theorem}[main]
Algorithm \ref{alg} gives a 0.595-approximation for any weighted input graph $G=(V,E,w)$.
\end{theorem}
\begin{proof}
Define
\begin{align*}
    &H_{ij} = (I - X_i X_j - Y_i Y_j - Z_i Z_j)/4, \\
    &H = \sum_{(i,j) \in E} w_{ij}H_{ij}. 
\end{align*}
We bound the expected energy of each case of Steps \ref{item:product_state} and \ref{item:matching} of the algorithm.
The subroutine of Step \ref{item:product_state} is directly from the main algorithm of \cite{GP19}, where in turn the analysis is borrowed from \cite{BOV14}.
More precisely, the energy of $\rho_1$ with respect to $H_{ij}$ is
\begin{align*}
    \Tr \rho_1 H_{ij}&= \frac{1}{16}\Tr [(I-X\otimes X - Y \otimes Y - Z \otimes Z) \\
    &((I + u_{i,1}X + u_{i,3}Y + u_{i,3}Z)\otimes (I + u_{j,1}X + u_{j,3}Y + u_{j,3}Z))] \\
    &=\frac{1}{4}(1- u_i\cdot u_j),
\end{align*}
and its expected value is 
\begin{align}
    \Tr \rho_1 H_{ij} &= \frac{1- f_3(v_i\cdot v_j)}{4} \label{eq:energy_prod}
\end{align}
where
\begin{align*}
    f_3(x) = \frac{2}{3}\left( \frac{\Gamma(2)}{\Gamma(1.5)}\right)^2 x _2F_1\left(\substack{1/2, 1/2 \\ \\2}; x^2  \right).
\end{align*}
The estimation is from Lemma 2.1 of \cite{BOV14}.

Now we turn to the analysis of the energy of $\rho_2$ from Step \ref{item:matching}. 
Note that if an edge $e$ is matched, then $\Tr \rho_2 H_{e}  = 1$, and if $e$ is not matched, then $\Tr \rho_2 H_{e} =1/4.$
More succinctly, $\Tr \rho_2 H_e = 1/4 + 3 M_e/4$.  Since the SDP values, $(h^+_e)_{e \in E}$ satisfy monogamy of entanglement on a star and triangle (Equations~\eqref{eq:hij_monogamy} and \eqref{eq:hij_triangle}), we have:
\begin{gather*}
 \sum_{j \in N(i)} 2h^+_{ij} \le 1, \text{for all } i \in V,\\
 2h^+_{ij} + 2h^+_{jk} + 2h^+_{ik} \leq 1, \text{for all } i,j,k \in V,
\end{gather*}
where the former correspond to the constraints~\eqref{eq:matching_vertex} of the LP~\eqref{eq:matching_lp}, and the latter correspond to the constraints~\eqref{eq:matching_set} with $|S|=3$.  Then by Lemma \ref{lem:matching_approx}, $(\frac{8}{5} h^+_e)_{e \in E}$ is feasible for the LP.  This implies that the optimal solution of the LP, namely a maximum weight matching, has weight at least that of $(\frac{8}{5} h^+_e)_{e \in E}$:
\begin{equation*}   
\sum_{e \in E} w_e M_e \ge \frac{8}{5} \sum_{e \in E} w_{e} h_{e}^+.
\end{equation*}
Therefore,
\begin{equation}
\sum_{e \in E} w_{e} \Tr \rho_2 H_{e} =  \sum_{e \in E} w_{e}\left(\frac{1}{4} + 
\frac{3}{4}M_e\right) \ge \sum_{e \in E} w_{e}\left(\frac{1}{4} +\frac{6}{5} h_e^+\right). \label{eq:energy_matching}
\end{equation}

By definition, $v(I)\cdot(v(I) - v(X_i X_j) - v(Y_i Y_j) -v(Z_i Z_j))/4 = (1 -3 v_i\cdot v_j)/4 = 1/2 + h_{ij}.$
So we have 
\begin{align}
    h_{ij} = -\frac{1+3v_i\cdot v_j}{4}.   \label{eq:hij}
\end{align} 
Let $\sigma$ be the density matrix of the output state of the algorithm.
By combining the energy estimation of the two cases \eqref{eq:energy_prod} and \eqref{eq:energy_matching}, we get,
\begin{align*}
    \sum_{(i,j) \in E} w_{ij} \Tr \sigma H_{ij} &=
    \max \left\{\sum_{(i,j)\in E} w_{ij} \rho_1 H_{ij}, \sum_{(i,j)\in E} w_{ij} \rho_2 H_{ij}\right\}\\
    &\ge \sum_{(i,j)\in E} w_{ij} \left[ p\frac{1 - f_3(v_i\cdot v_j)}{4} + (1-p) \left(\frac{1}{4} + \frac{6}{5}\left(-\frac{1 + 3 v_i\cdot v_j}{4}\right)^+\right)\right],
\end{align*}
for any $p \in [0,1]$. Since $-1\le v_i\cdot v_j \le 1/3$, it suffices to find
\begin{align*}
     \alpha &:= \max_{p\in [0,1]} \min_{x \in [-1, 1/3]} \left[ p\frac{1 - f_3(x)}{4} + (1-p) \left(\frac{1}{4} + \frac{6}{5}\left(-\frac{1 + 3x}{4}\right)^+ \right)  \right]  \left. \middle/~ {\frac{ 1- 3x}{4}} \right. \\
     &= 0.595,
\end{align*}
where the maximum occurs at $p = 0.674$. This proves the theorem since
\begin{align*}
    \sum_{(i,j) \in E} w_{ij} \Tr \sigma H_{ij}   &\ge \sum_{(i,j)\in E} w_{ij} \left[ p_0\frac{1 - f_3(v_i\cdot v_j)}{4} + (1-p_0) \left(\frac{1}{4} + \frac{6}{5}\left(-\frac{1 + 3 v_i\cdot v_j}{4}\right)^+\right)\right] \\
     &\ge 0.595\sum_{(i,j)\in E} w_{ij} \frac{ 1-3v_i \cdot v_j}{4} \ge 0.595\lambda_{max} (H).
\end{align*}
\end{proof}

\begin{remark}
The present approach can be likely be improved by deriving analogues of Corollary~\ref{cor:triangle_bound} for larger odd-sized sets of qubits.  This would enable a stronger version of Lemma~\ref{lem:matching_approx} with $\alpha > \frac{4}{5}$.  However, the best approximation ratio achievable by such approaches, corresponding to $\alpha = 1$, is $0.606$. This is also the approximation ratio of our algorithm on bipartite graphs, since in this case LP~\eqref{eq:matching_lp} gives the value of a maximum weight matching even when constraints~\eqref{eq:matching_set} are absent from the LP.
\end{remark}

\bigskip
\section*{Open problems}
Understanding the approximability of QMC is likely to bring a deeper understanding of the more general local Hamiltonian problem, just as resolving the approximability of Max Cut (up to the Unique Games Conjecture) has had surprising consequences for the theory of classical constraint satisfaction problems.  We list relevant research directions below.
\begin{itemize}
    \item Find a rigorous quantum approximation algorithm for QMC.
    \item Find a heuristic quantum algorithm (e.g.~VQE-based) for QMC that outperforms rigorous classical algorithms.
    \item The approximability of QMC using product states (i.e.~tensor products of 1-qubit states) is well understood~\cite{PT22,HNPTW21}. Find the best approximation ratio achievable using a tensor product of 1- and 2-qubit states. Can this be obtained using matchings? Which level of the lassere hierarchy is necessary to achieve this? 
    \item Find an entanglement convexgamy relation (i.e.~tighter non-linear inequalities on star graphs) from a valid level-$k$ SDP solution on $d$ edges. Does the optimal such relation (i.e.~one precisely describing the feasible space) arise at a constant level $k$ (with respect to $d$)?
\end{itemize}

\bigskip
\section*{Acknowledgements}
Eunou Lee appreciates helpful comments on the write-up by Andrus Giraldo.  
Eunou Lee was supported by Individual Grant No.CG093801 at Korea Institute for Advanced Study. 

This work is supported by a collaboration between the US DOE and other Agencies.
Ojas Parekh was supported by the U.S. Department of Energy (DOE), Office of Science, National Quantum Information Science Research Centers, Quantum Systems Accelerator.  Additional support is acknowledged from DOE, Office of Science, Office of Advanced Scientific Computing Research, Accelerated Research in Quantum Computing, Fundamental Algorithmic Research in Quantum Computing.

This article has been authored by an employee of National Technology \& Engineering Solutions of Sandia, LLC under Contract No. DE-NA0003525 with the U.S. Department of Energy (DOE). The employee owns all right, title and interest in and to the article and is solely responsible for its contents. The United States Government retains and the publisher, by accepting the article for publication, acknowledges that the United States Government retains a non-exclusive, paid-up, irrevocable, world-wide license to publish or reproduce the published form of this article or allow others to do so, for United States Government purposes. The DOE will provide public access to these results of federally sponsored research in accordance with the DOE Public Access Plan \url{https://www.energy.gov/downloads/doe-public-access-plan}.

\bibliography{ref}

\end{document}